\documentclass[12pt]{article}
\usepackage{amsfonts,amsmath,amssymb,amsthm,mathrsfs,url,hyperref}

\title{Uniform Probability Distribution Over All Density Matrices}
\author{
Eddy Keming Chen\footnote{Department of Philosophy, University of California San Diego, 9500 Gilman Dr, La Jolla, CA 92093, USA.  Email: eddykemingchen@ucsd.edu}~~and
Roderich Tumulka\footnote{Fachbereich Mathematik,
     Eberhard-Karls-Universit\"at, Auf der Morgenstelle 10, 72076
     T\"ubingen, Germany. Email: roderich.tumulka@uni-tuebingen.de}
}
\date{April 2, 2020}

\newcommand{\Hilbert}{\mathscr{H}}
\newcommand{\Kilbert}{\mathscr{K}}

\newcommand{\sD}{\mathscr{D}}

\newcommand{\sP}{\mathscr{P}}
\newcommand{\sS}{\mathscr{S}}
\newcommand{\sT}{\mathscr{T}}
\newcommand{\sU}{\mathscr{U}}
\newcommand{\sV}{\mathscr{V}}
\newcommand{\sW}{\mathscr{W}}

\newcommand{\EEE}{\mathbb{E}}

\newcommand{\RRR}{\mathbb{R}}

\newcommand{\SSS}{\mathbb{S}}
\newcommand{\NNN}{\mathbb{N}}
\newcommand{\scp}[2]{\langle #1|#2 \rangle}

\DeclareMathOperator{\tr}{tr}

\newcommand{\vc}{\boldsymbol{c}}

\newcommand{\vw}{\boldsymbol{w}}

\newcommand{\vmu}{\boldsymbol{\mu}}
\newcommand{\vnu}{\boldsymbol{\nu}}

\theoremstyle{plain}
\newtheorem{prop}{Proposition}
\newcommand{\be}{\begin{equation}}
\newcommand{\ee}{\end{equation}}
\newcommand{\dm}{\rho}

\begin{document}
\maketitle
\begin{abstract}
Let $\Hilbert$ be a finite-dimensional complex Hilbert space and $\sD$ the set of density matrices on $\Hilbert$, i.e., the positive operators with trace 1. Our goal in this note is to identify a probability measure $u$ on $\sD$ that can be regarded as the uniform distribution over $\sD$. We propose a measure on $\sD$, argue that it can be so regarded, discuss its properties, and compute the joint distribution of the eigenvalues of a random density matrix distributed according to this measure.

\medskip

  \noindent 
  Key words: random matrix; finite-dimensional Hilbert space.
\end{abstract}

\section{Introduction}

With every probability distribution $\mu$ over wave functions, i.e., over the unit sphere $\SSS(\Hilbert)$ in a complex Hilbert space $\Hilbert$, there is associated a density matrix
\be
\dm = \int_{\SSS(\Hilbert)}\mu(d\psi) \, |\psi\rangle\langle\psi|\,.
\ee 
In this note, in contrast, we consider a probability distribution over density matrices, and we ask whether there exists a distribution that should be regarded as the \emph{uniform} distribution $u$ over all density matrices. Our considerations involve certain applications of random matrix theory.

Here is our motivation for the question. Density matrices can arise not only as encoding random $\psi$'s, but also as partial traces of states of larger systems; moreover, it is conceivable that even the fundamental state, as known to nature, is a density matrix $\dm$. For example, it is easy to set up a version of Bohmian mechanics in which the particles are guided, not by a wave function $\psi$, but by a density matrix $\dm$ \cite{DGTZ05}.  The density matrix, in such a theory, is not an expression of our ignorance of the actual pure state even if we continue to call it a ``mixed state,'' nor is it an expression of entanglement with another system, but it is a fundamental object---a physical variable on which the motion of the Bohmian particles depends. 
Likewise, a density matrix can be a fundamental object in collapse theories or many-worlds theories \cite{AGTZ13}. But if a density matrix is a fundamental object in nature, then it makes sense to consider  a random density matrix. For example, when considering the initial state of the universe, it is common to consider a random state in a particular subspace $\Hilbert_{PH}$ of the Hilbert space of the universe associated with very low entropy; the statement that the initial state of the universe lies in the subspace $\Hilbert_{PH}$ is often called the \emph{past hypothesis} (PH) \cite{Alb}. Here, one usually has in mind a random pure state $\psi$ in $\Hilbert_{PH}$, but if states $\dm$ that are fundamentally mixed are possible, as illustrated by the above-mentioned versions of Bohmian mechanics, collapse theories, and many-worlds theories, we can also consider a random $\dm$ in $\Hilbert_{PH}$ \cite{Ch19}. Since one considers for $\psi$ the uniform distribution over $\SSS(\Hilbert_{PH})$, the analog would involve the uniform distribution $u$ over all $\rho$ concentrated in $\Hilbert_{PH}$, which brings us to the question whether such a distribution $u$ exists, whether it is uniquely defined, and what it looks like. 

Here, we propose a natural definition of $u$ on any Hilbert space $\Hilbert$ of finite dimension $d\in\NNN$. It will be clear from the definition that $u$ exists and is unique. For infinite-dimensional Hilbert spaces $\Hilbert$, it does not seem that a uniform distribution exists over the density matrices on $\Hilbert$, which is not surprising as there is no uniform distribution either over $\Hilbert$ itself or $\SSS(\Hilbert)$. Of course, our reasoning also yields, for any subspace $\Hilbert$ of a bigger Hilbert space $\Kilbert$ with $\dim\Hilbert<\infty$, a uniform probability distribution over the density matrices concentrated on $\Hilbert$, regardless of whether $\dim\Kilbert$ is finite or infinite. We show that $u$ is invariant under unitary operators on $\Hilbert$, so that a $u$-distributed $\rho$ has an eigenbasis that is uniformly distributed in the set of all orthonormal bases of $\Hilbert$. Furthermore, we compute the joint distribution of the eigenvalues of $\rho$. The expectation value of $\rho$ is $d^{-1}I$, where $I$ is the identity operator on $\Hilbert$.

In applications, the normalized measure $u$ may often play the role of a \emph{typicality measure} (see, e.g., \cite[Sec.~6]{GLTZ10} and \cite[Sec.~7.1]{GLTZ20}) rather than that of direct probability. That is, it may serve for defining what is true of \emph{most} density matrices (that are, say, concentrated in a certain subspace such as $\Hilbert_{PH}$). For example, the properties of $u$ will entail that for a bipartite system, most density matrices are entangled, just as most pure states are \cite{GLTZ06}.

Concerning the past hypothesis, another approach proposes to take the initial density matrix of the universe to be the normalized projection onto $\Hilbert_{PH}$ \cite{Ch18,Ch20}. So, one could consider different kinds of initial conditions: a random $\psi$ with uniform distribution over $\SSS(\Hilbert_{PH})$, a fixed density matrix proportional to the projection to $\Hilbert_{PH}$, or a random density matrix with distribution $u$ over the density matrices in $\Hilbert_{PH}$. It seems reasonable to expect that all three theories are empirically equivalent, according to the appropriate sense of typicality.  We leave that issue to another paper.

\section{Definition of the Measure}

Let $\sS$ be the space of self-adjoint operators on $\Hilbert$ (a real vector space of dimension $d^2$), $\sP\subset\sS$ the set of positive operators on $\Hilbert$, and $\sT_c$ the set of self-adjoint operators with trace $c$ (an affine subspace of $\sS$ of dimension $d^2-1$); the set $\sD$ of all density matrices is $\sD=\sP\cap \sT_1$. Since for $d=1$, $\sD$ has only one element, we assume $d\geq 2$. Let $\sP^\circ$ denote the interior of $\sP$, which is the set of positive definite operators on $\Hilbert$, and $\sD^\circ=\sP^\circ \cap \sT_1$ the interior of $\sD$ in $\sT_1$ (the set of density matrices for which 0 is not an eigenvalue).

As in every affine space of finite dimension, there is a natural notion of volume in $\sT_1$: a nonzero translation-invariant measure on the Borel $\sigma$-algebra of $\sT_1$. It is well known that this measure is unique up to a global positive factor.

\begin{prop}
For every such measure, the volume of $\sD$ is neither zero nor infinite.
\end{prop}

\begin{proof}
It is not zero because the interior $\sD^\circ$ is open and non-empty. That it is finite will follow once we show that $\sD$ is compact and therefore bounded in $\sT_1$. The compactness of $\sD$ will follow from the fact that the continuous image of any compact set is compact. Here, the relevant mapping is $\varphi: \RRR^d\times U(d) \to \sS$ (where $U(d)$ denotes the unitary group of $d\times d$ matrices, here regarded as orthonormal bases of $\Hilbert$) defined by
\be\label{phidef}
\varphi(\lambda_1,\ldots,\lambda_d,\psi_1,\ldots,\psi_d) = \sum_{i=1}^d \lambda_i \psi_i\,.
\ee
$\varphi$ is clearly continuous, and since $U(d)$ is known to be compact and 
\be\label{Lambdadef}
\Lambda:=\Bigl\{(\lambda_1,\ldots,\lambda_d)\in[0,1]^d:
\lambda_1\geq \ldots \geq \lambda_d\,,~\sum_{i=1}^d \lambda_i=1 \Bigr\}
\ee
is clearly compact, also $\Lambda\times U(d)$ is compact,
which gets mapped to $\sD$.
\end{proof}

Thus, one can restrict the volume measure in $\sT_1$ to $\sD$ and normalize, which removes the arbitrary constant. The resulting measure is the desired measure $u$.

\section{Properties of the Measure}

\subsection{Unitary Invariance}

\begin{prop}\label{prop:unitary}
$u$ is invariant under unitary transformations $U$ of $\Hilbert$. 
\end{prop}

\begin{proof}
$U$ maps $\sS$ to itself in a linear way and maps $\sT_1$ to itself in an affine-linear way. Thus, any translation invariant measure on $\sT_1$ will be mapped by $U$ to a multiple of itself. Since $U$ also maps $\sP$ to itself, it also maps $\sD$ to itself. As a consequence, it must preserve volumes when acting on $\sT_1$, and so it preserves $u$.
\end{proof}

\begin{proof}[Alternative proof.]
Equip $\sS$ with the Hilbert-Schmidt inner product 
\be\label{HS}
\langle A,B \rangle = \tr(AB), 
\ee
which is invariant under $U$. Using the inner product, one has a notion of area on every surface, in particular on $\sT_1$. $u$ is just the normalized surface area restricted to $\sD$, and it follows that surface area is invariant under $U$.
\end{proof}

Note that unitary invariance does not uniquely select the measure $u$. Unitary invariance means that the joint distribution of the eigenvectors of $\rho$ is uniform while saying nothing about the joint distribution of the eigenvalues. The property that selects $u$ as the natural normalized measure on $\sD$ is that $u$ is, when looked at in the right way, just volume.

\subsection{Expectation and Covariance}

The \emph{covariance} of a random vector $V$ in a real vector space $\sV$ with inner product $\langle~,\,\rangle$ is defined to be the operator $C:\sV\to\sV$ such that
\be
\langle v,Cv'\rangle = \EEE \Bigl[\bigl\langle v,(V-\EEE V)\bigr\rangle \bigl\langle(V-\EEE V),v'\bigr\rangle \Bigr] 
\ee
for all $v,v'\in\sV$.

\begin{prop}
A $u$-distributed $\rho$ has expectation
\be\label{Erho}
\EEE\rho = \tfrac{1}{d}I
\ee 
and covariance (in $\sV=\sS$ with Hilbert-Schmidt inner product \eqref{HS}) 
\be
C= c(d)\,P_{\sT_0}\,,
\ee
with $c(d)>0$ some constant\footnote{After completion of this paper, we have become aware of results of Tucci \cite{Tu02} that imply that $c(d)=\frac{1}{d(d^2+1)}$.} and $P_{\sT_0}$ the projection to the set $\sT_0$ of traceless operators in $\sS$.
\end{prop}

\begin{proof}
As a consequence of Proposition~\ref{prop:unitary}, $\EEE\rho$ must be invariant under $U$, and since the only operators in $\Hilbert$ invariant under all unitaries are the multiples of the identity $I$, \eqref{Erho} follows.

Likewise, $C$ must be invariant under $U(d)$. To determine all $U(d)$-invariant   operators on $\sS$, we first show that the representation of $U(d)$ on $\sS$ is the direct sum of two irreducible representation spaces, $\RRR I$ (the multiples of the identity) and $\sT_0$. 

Clearly, $\RRR I$ and $\sT_0$ are $U(d)$-invariant (as $\tr (UAU^{-1})=\tr(A)$), they are orthogonal in the Hilbert-Schmidt inner product, their sum is $\sS$, and $\RRR I$ is irreducible because it is 1-dimensional. In order to show that $\sT_0$ is irreducible, we show that $\{0\}$ and $\sT_0$ are its only invariant subspaces. To this end, let $\sU\neq \{0\}$ be an invariant subspace of $\sT_0$; we show that $\sU+\RRR I=\sS$, which implies that $\sU=\sT_0$. Note that $\sU+\RRR I$ is invariant. Let $0\neq A\in \sU$. Then $A$ has at least two different eigenvalues; choose an orthonormal basis of $\Hilbert$ that diagonalizes $A$. We show that all $B\in\sS$ that are diagonal in the same basis also lie in $\sU+\RRR I$; it then follows by applying unitaries that $\sU+\RRR I=\sS$. For this, it suffices to show that for $d\geq 2$ the only subspace of $\RRR^d$ that is invariant under permutation of components and contains $\vc:=(1,1,\ldots,1)$ and some vector not proportional to $\vc$ is $\RRR^d$ itself. Indeed, if $\sW$ is such a subspace and $\vw\in \sW\setminus \RRR \vc$, then $w_i\neq w_j$ for some $i\neq j$. Let $\vw'$ be the vector obtained from $\vw$ by permuting $w_i$ and $w_j$, then $\vw'':=\vw-\vw' \in \sW$ has $w''_i= w_i-w_j$, $w''_j=w_j-w_i$, while all other components of $\vw''$ vanish. Thus, using permutations again, $(1,-1,0,\ldots,0)\in\sW$ and
\begin{multline}
(1,0,\ldots,0) =\\ \tfrac{1}{d}\Bigl[\vc + (1,-1,0,0,\ldots,0) + (1,0,-1,0,\ldots,0) + \ldots + (1,0,\ldots,0,-1)\Bigr] \in \sW\,.
\end{multline}
By permutation, all $(0,\ldots,0,1,0,\ldots,0)\in\sW$, so $\sW=\RRR^d$.

Now, since $\sT_0$ is irreducible, we can apply Schur's lemma \cite{Schur}. Since the irreducible representations $\RRR I$ (which has dimension 1) and $\sT_0$ (which has dimension $d^2-1\geq 3$) are inequivalent, Schur's lemma yields that every $U(d)$-invariant operator $C:\sS\to\sS$ is of the form
\be
C= \tilde c P_{\RRR I} + c P_{\sT_0}\,.
\ee
For the covariance operator $C$, since always $\rho-\EEE\rho \in \sT_0$, we have that $\tilde c=0$. 
\end{proof}

We can characterize the value of $c=c(d)$ as follows. Fix $\psi\in\SSS(\Hilbert)$ and set $v=v'=|\psi\rangle\langle\psi|$. Then
\begin{align}
\langle v,Cv\rangle 
&= \EEE\Bigl[\bigl(\tr [v(\rho-\EEE\rho)]\bigr)^2\Bigr]\\
&= \EEE\Bigl[\bigl(\scp{\psi}{\rho|\psi}-d^{-1}\bigr)^2\Bigr]\\
&= \EEE\Bigl[\scp{\psi}{\rho|\psi}^2\Bigr] -2d^{-1}\EEE\scp{\psi}{\rho|\psi}+d^{-2}\\
&= \EEE\Bigl[\scp{\psi}{\rho|\psi}^2\Bigr] -d^{-2}\,.
\end{align}
On the other hand,
\begin{align}
\langle v,Cv\rangle 
&= c(d) \langle v,P_{\sT_0}v\rangle \\
&= c(d) \Bigl(\langle v,v\rangle - \langle v,P_{\RRR I}v\rangle \Bigr)\\
&= c(d) \Bigl(1 - \langle v,d^{-1/2}I\rangle \langle d^{-1/2}I, v\rangle \Bigr)\\[1mm]
&= c(d) \bigl(1 - d^{-1}(\tr v)^2 \bigr)\\[2mm]
&= c(d) (1-d^{-1})\,.
\end{align}
Thus,
\be
c(d) = \tfrac{d}{d-1} \EEE\Bigl[\scp{\psi}{\rho|\psi}^2\Bigr] - \tfrac{1}{d(d-1)}\,.
\ee
We did not succeed in evaluating the expectation value.\footnote{Tucci \cite{Tu02} showed that $\EEE\Bigl[\scp{\psi}{\rho|\psi}^2\Bigr] = \frac{d+1}{d(d^2+1)}$, which leads to the formula of Footnote 1.}

\subsection{Distribution of Eigenvalues}

Let $T_1$ be the plane
\be\label{T1def}
T_1:= \Bigl\{(\lambda_1,\ldots,\lambda_d)\in\RRR^d: \sum_{i=1}^d \lambda_i=1  \Bigr\}\,.
\ee

\begin{prop}
Under $u$, the eigenvalues $\lambda_1\geq \lambda_2 \geq \ldots \geq \lambda_d$ of $\rho$ have joint distribution in $\Lambda\subset T_1$ with density
\be\label{eig}
f(\lambda_1,\ldots,\lambda_d) = \mathcal{N} \prod_{1\leq i<j \leq d} |\lambda_i-\lambda_j|^2
\ee
relative to the volume measure in $T_1$ with normalization constant $\mathcal{N}>0$.
\end{prop}

\begin{proof}
The strategy of proof is to use, instead of volume on $\sS$, a Gaussian unitary ensemble, for which the distribution of the eigenvalues is known, and then let its variance tend to infinity, so that the distribution becomes flat on every compact set. 

The Gaussian unitary ensemble \cite{wiki} is the probability distribution over self-adjoint $d\times d$ matrices $X_{ij}= A_{ij}+iB_{ij}$ with real part $A_{ij}=A_{ji}$ and imaginary part $B_{ij}=-B_{ji}$ such that all $A_{ij}$ ($i\leq j$) and all $B_{ij}$ ($i<j$) are independent random variables, where $A_{ij}$ with $i<j$ and $B_{ij}$ are Gaussian with mean 0 and variance $1/(2d)$, while the $A_{ii}$ are Gaussian with mean 0 and variance $1/d$. Thus, the joint distribution of all $X_{ij}$ has density (with lower case symbols the possible values of random variables)
\begin{align}
f_X(x_{11},x_{12},\ldots,x_{dd})
&\propto \prod_{i<j} e^{-da_{ij}^2}e^{-db_{ij}^2} \prod_i e^{-da_{ii}^2/2}\\
&=  \prod_{i,j=1}^d e^{-d|x_{ij}|^2/2}\\
&=  e^{-d \tr x^2/2}\,.
\end{align}
It is known \cite{wiki} that the eigenvalues $\mu_1\geq \ldots\geq\mu_d$ of $X$ have joint distribution with density
\be
g_X(\mu_1,\ldots,\mu_d) \propto \prod_{k=1}^d e^{-\frac{d}{2} \mu_k^2} \prod_{1\leq i<j \leq d} |\mu_i-\mu_j|^2\,.
\ee
That is, $\varphi^{-1}$ maps the distribution $f_X(x) \, dx$ to the product of $g_X(\vmu)\,d\vmu$ (with $\vmu=(\mu_1,\ldots,\mu_d)$) and the uniform distribution on $U(d)$.

Now consider $Y:=\sigma X$ with arbitrary $\sigma>0$ that we will ultimately let tend to infinity. $Y$ has density
\be\label{fY}
f_Y(y_{11},y_{12},\ldots,y_{dd}) \propto e^{-d \tr y^2/2\sigma^2}\,,
\ee
and its eigenvalues $\nu_1=\sigma \mu_1,\ldots,\nu_d=\sigma \mu_d$ have joint density
\be\label{gY}
g_Y(\nu_1,\ldots,\nu_d) \propto \prod_{k=1}^d e^{-d \nu_k^2/2\sigma^2} \prod_{1\leq i<j \leq d} \frac{|\nu_i-\nu_j|^2}{\sigma^2}\,.
\ee
Again, $\varphi^{-1}$ maps the distribution $f_Y(y) \, dy$ to the product of $g_Y(\vnu)\, d\vnu$ (with $\vnu=(\nu_1,\ldots,\nu_d)$) and the uniform distribution on $U(d)$.

Since $\varphi$ maps $T_1\times U(d)$ to $\sT_1$, it maps the conditional distribution of $\vnu$ on $T_1$, times the uniform distribution on $U(d)$, to the conditional distribution of $Y$ on $\sT_1$. Likewise, it maps the conditional distribution of $\vnu$ on $\Lambda$, times the uniform distribution on $U(d)$, to the conditional distribution of $Y$ on $\sD$. Note that the conditional distribution of $Y$ on $\sT_1$ has density, up to a normalizing factor, given by $f_Y$ restricted to $T_1$, and the conditional distribution of $\vnu$ on $T_1$ has density $g_Y$ on $T_1$ up to a factor. In the limit $\sigma\to\infty$, the right-hand side of \eqref{fY} converges to 1, in fact uniformly on the compact set $\sD$; thus, also $f_Y$ (including the appropriate normalizing factor) converges uniformly to 1 on $\sD$. On the other hand, in the same way, the right-hand side of \eqref{gY}, after dropping the factors of $\sigma$ in the denominator, converges to $\prod |\nu_i-\nu_j|^2$, in fact uniformly on the compact set $\Lambda$. We want to draw the conclusion that $\varphi$ maps the limit of $g_Y$-conditional-on-$\Lambda$ (times the uniform distribution on $U(d)$) to the limit of $f_Y$-conditional-on-$\sD$ (i.e., to $u$).

To justify this conclusion, we note the following. The interior of $\Lambda$ is
\be
\Lambda^\circ = \Bigl\{(\lambda_1,\ldots,\lambda_d)\in T_1: \lambda_1>\ldots > \lambda_d>0 \Bigr\}\,.
\ee
Since $\Lambda$ is a convex set, its boundary has measure zero in $T_1$; thus, it does not matter whether we consider continuous measures on $\Lambda$ or $\Lambda^\circ$. For eigenvalues in $\Lambda^\circ$, the orthonormal basis of eigenvectors is unique up to phases; that is, $\varphi$ maps $\Lambda^\circ \times [U(d)/U(1)^d]$ bijectively to the set of non-degenerate positive definite density matrices, a dense set of full $u$-measure in $\sD$. Since $\varphi$ is smooth (in particular) on $T_1\times U(d)$, so is its Jacobian determinant; since $\Lambda \times U(d)$ is compact, the Jacobian is bounded on $\Lambda\times U(d)$. According to the transformation formula for integrals, the density of the pre-image is the Jacobian times the density of the image; as a consequence, if the Jacobian is bounded and the density of the image converges uniformly, then so does the density of the pre-image. That is, we can pull the limit through $\varphi$, as we claimed.

The upshot is that $\varphi^{-1}$ maps $u$ to
\begin{align}
&\lim_{\sigma\to\infty} g_Y(\vnu) \, d\vnu \times \mathrm{uniform}_{U(d)}\\
&= \mathcal{N} \biggl(\prod_{1\leq i<j \leq d} |\nu_i-\nu_j|^2\biggr) d\vnu \times \mathrm{uniform}_{U(d)}\,,
\end{align}
which proves \eqref{eig} (and by the way again the unitary invariance of $u$).
\end{proof}

\textit{Note added.} After completion of this paper we have learned of prior works \cite{H98,ZS01,Tu02,ZS03,ZS04} that considered the measure we denote by $u$. Hall \cite{H98} asked which distribution over the density matrices ``corresponds to minimal prior knowledge'' or is ``most random.'' That is perhaps the same as asking which distribution is uniform, or perhaps it is subtly different. He came up with three proposed answers, one of which is $u$, and regarded another one as ``most random.'' Hall also arrived at the formula \eqref{eig} for the distribution of the eigenvalues, but in a different way than we did. {\.Z}yczkowski and Sommers computed \cite{ZS03} the volume of $\sD$ in $\sT_1$ according to the Hilbert-Schmidt metric (and thus the normalization constant in the definition of $u$), and computed \cite[Eq.~(3.7)]{ZS01} the normalization constant in Proposition 4 to be  $\mathcal{N}= (d^2-1)!/\prod_{k=1}^d [k!(k-1)!]$. They also showed \cite{ZS01} that for a uniformly random unit vector in $\Hilbert\otimes \Hilbert$, the reduced density matrix in $\Hilbert$ is $u$-distributed, and that, for a random $d\times d$ matrix $A$ from the Ginibre ensemble (i.e., for which each entry is independent complex Gaussian with mean 0 and variance 1), $\rho:=AA^*/\tr(AA^*)$ has distribution $u$; asymptotics for large $d$ are studied in \cite{ZS04}. Tucci \cite{Tu02} also considered $u$, called it the ``uniform ensemble of density matrices,'' and computed all moments of the entries of a $u$-distributed $\rho$.

\bigskip\bigskip

\noindent{\it Acknowledgments.} We thank Stefan Keppeler for helpful discussion and Michael Hall, Christian Majenz, Ion Nechita, Michael Walter, and Karol {\.Z}yczkowski for pointing to relevant literature.

\end{document}